\documentclass[12pt, a4paper]{article}

\usepackage{amsfonts, latexsym, amssymb, amsthm, amsmath, mathrsfs, esint}

\usepackage{enumerate}

\newcommand{\p}{\varphi}
\newcommand{\la}{\lambda}
\renewcommand{\a}{\alpha}
\renewcommand{\b}{\beta}
\newcommand{\D}{\nabla}

\def\d{\partial}

\def\weak{\rightharpoonup}
\def\RR{\mathbb{R}}
\renewcommand{\t}{\widetilde}
\newcommand{\R}{\mathbb{R}}

\def\LL{\mathcal{L}}
\def\AA{\mathcal{A}}

\def\DD{\mathcal{D}}
\def\W{\mathcal{W}}
\def\WW{\mathcal{W}}

\newcommand{\SM}{\mathcal{S}M}

\def\D{\nabla}
\def\a{\alpha}
\def\b{\beta}

\renewcommand{\o}{\overline}
\newcommand{\al}[1]{\left\langle #1 \right\rangle}
\renewcommand{\div}{\mathrm{div\ }}

\newtheorem{theorem}{Theorem}[section]
\newtheorem{lemma}[theorem]{Lemma}
\newtheorem{proposition}[theorem]{Proposition}
\newtheorem{corollary}[theorem]{Corollary}
\newtheorem{definition}[theorem]{Definition}
\newtheorem{remark}[theorem]{Remark}

\DeclareMathOperator{\Tr}{Tr}

\DeclareMathOperator{\cof}{cof}
\DeclareMathOperator{\Def}{Def}
\DeclareMathOperator{\End}{End}
\DeclareMathOperator{\Hess}{Hess}

\begin{document}
\title{Nonlinear bending theories\\
for non Euclidean plates}

\author{
{\sc Peter Hornung}
\footnote{
    Work supported by the DFG.
}
}

\date{}

\maketitle

\begin{abstract}
Thin growing tissues (such as plant leaves) can be modelled by a bounded
domain $S\subset\RR^2$ endowed with a Riemannian
metric $g$, which models the internal strains caused by the
differential growth of the tissue.
The elastic energy is given by a nonlinear isometry-constrained 
bending energy functional
which is a natural generalization of Kirchhoff's plate functional.
We introduce and discuss a natural notion of (possibly non-minimising) stationarity 
points. We show that rotationally symmetric immersions of the unit disk are
stationary, and we give examples of metrics $g$ leading to functionals with
infinitely many stationary points.
\end{abstract}

\section{Introduction}\label{Sec0}

Thin growing tissues in biology, such as plant leaves, display bending patterns even if no external forces are applied.
A similar pattern formation is observed in thin elastic sheets that have undergone plastic deformation, cf. \cite{SharonMarderSwinney}.
What garbage bags and growing tissues have in common is that their reference configuration is not stress-free.
In the case of growing tissues, the internal stresses arise from the differential growth of the leaf: typically,
leaves do not grow much near their midrib, but their cells keep dividing and growing to their target size 
near the edges of the leaf. The excess length created by this
growth near the leaf edge creates internal stresses in the planar reference configuration. These
stresses can be relaxed by bending the leaf out of the plane.
\\
Such questions have been studied in \cite{EfratiSharonKupferman, GemmerVenka} and elsewhere.
Two very different scenarios can arise: either there is a way of fully relaxing the internal stresses by deforming the leaf (with finite
bending energy), or this is not possible. In this paper we are interested in the former case.
\\
The reference configuration of the leaf is modelled by a bounded domain $S\subset\R^2$ and the inhomogeneous and anisotropic
local growth pattern is modelled by a Riemannian `target' metric $g$ on $S$.
The fact that there is a way of fully relaxing the internal stresses by deforming the leaf with finite
bending energy means precisely that the set 
$$
W^{2,2}_g(S) = \left\{u\in W^{2,2}(S, \RR^3) : (\D u)^T(\D u) = g
\mbox{ almost everywhere in }S\right\}
$$
of $W^{2,2}$ isometric immersions of
$(S, g)$ into $\RR^3$ is nonempty. If this is the case, then the asymptotic 
behaviour of the three dimensional elastic energy of the leaf with
small thickness is captured by the natural generalization of Kirchhoff's 
nonlinear bending theory of plates which we introduce next.
Firstly, for any regular immersion $u : S\to\R^3$ we introduce the
Willmore functional (cf. e.g. \cite{Weiner, Schatzle-CV})
\begin{equation}
\label{willm}
W(u) = \frac{1}{4}\int_{S} H^2 d\mu_g + \int_{\d S} \kappa_g d\mu_{g_{\d}},
\end{equation}
where $\kappa_g$ is the geodesic curvature of $\d S$
and $g$ is the metric induced by $u$, the induced area measure on $S$ is $\mu_g$
and the induced boundary measure on $\d S$ is $\mu_{g_{\d}}$.
The Willmore functional has been studied extensively in the literature,
cf. \cite{Willmore-book, Simon-willmore, PinkallSterling} and the references cited therein. Its analytical properties
have been systematically studied in \cite{KS-flow, BauerKuwert, KS-annals, Riviere-willmore}.
More recently, there has been growing interest in constrained versions of the Willmore functional. The typical constraints
include prescribed conformal class (cf. \cite{bohle, KS-conformal, Riviere-L2}) or
fixed area and enclosed volume, cf. \cite{Helfrich, Schygulla}.

The relevant case for thin film elasticity is the restriction of the Willmore functional
\eqref{willm} to {\em isometric} immersions of the given Riemannian manifold $(S, g)$ into $\RR^3$.
More precisely, from now on $S\subset\RR^2$ will denote a bounded simply connected domain with smooth boundary,
and $g : \overline{S}\to\RR^{2\times 2}$ will be a given smooth Riemannian metric on $\overline{S}$. We will study
the restriction of the Willmore functional to the class $W^{2,2}_g(S)$.
That is, we will study the generalized Kirchhoff plate functionals
$$
\widetilde{\W}_g(u) =
\begin{cases}
W(u) &\mbox{ if }u\in W^{2,2}_g(S)
\\
+ \infty &\mbox{ otherwise.}
\end{cases}
$$
In addition to their key role in the modelling of thin films in nonlinear elasticity, 
these functionals are also entirely natural from a geometric viewpoint, 
as they are the simplest {purely extrinsic} functionals on
surfaces. It is a key feature of thin films in nonlinear elasticity that they
undergo large deformations with low energy. In contrast to von-K\'arm\'an theories
(cf., e.g., \cite{H-PRSE}), the above functionals admit
such large deformations. In fact, they arise naturally as (rigorous) asymptotic
theories from fully nonlinear three-dimensional elasticity in a bending energy regime, cf. 
\cite{FJMM-cras, KupfermanSolomon}.

In the case when $g$ is the standard flat metric, i.e. $g_{ij} = \delta_{ij}$,
the corresponding constrained Willmore functional $\widetilde{\W}_g$ agrees with the energy functional
arising Kirchhoff's nonlinear bending theory for thin elastic plates, cf. \cite{fjm1}. It was studied
in \cite{H-CPAM}. However, the arguments used in that paper heavily depend upon
the special structure (developability) of intrinsically flat surfaces, so they do not carry over to other metrics.
Moreover, they are not suited to formulate the concept of `stationary point' for functionals such as $\W_g$.
A possible notion was introduced in \cite{H-CPAM}
by regarding solutions to the Euler-Lagrange equations
derived in that paper as stationary points. 
But this definition does not help in 
defining stationarity for $\widetilde{\W}_g$ when $g$ is Riemannian metric other than the 
standard flat metric.
The need for a notion of (possibly non-minimising) stationary points
was pointed out e.g. in \cite{GemmerVenka-PhysicaD}.
\\
Moreover, from a conceptual viewpoint, a major drawback of the derivation of the Euler-Lagrange 
equation in \cite{H-CPAM} (which
is a system of ordinary differential equations) is that the relation to the 
classical Willmore equation remains unclear. So does the relation to a formal
Lagrange multiplier rule.

In this paper we provide a framework
for the analysis of isometry-constrained functionals such as $\widetilde{\W}_g$, 
which (as far as it goes) works for arbitrary Riemannian metrics $g$. 
This approach overcomes the problems mentioned earlier:
It leads to a natural notion of stationarity, the relation to the
classical Willmore equation is clear, 
and it leads to a natural formulation in terms of Lagrange multipliers.
\\
Our notion of stationarity is based on the geometric concept of bendings. 
Bendings are deformations of a given surface which preserve the metric. The velocity field $\tau$
of such a deformation is a solution of the system
\begin{equation}
\label{IB00}
\d_i u\cdot\d_j\tau + \d_j u\cdot\d_i\tau = 0\mbox{ for all }i = 1, 2,
\end{equation}
which is obtained by linearizing the system
$\d_i u\cdot\d_j u = g_{ij}$. Solutions $\tau$ to \eqref{IB00} are called infinitesimal bendings
of $u$.
\\
Our abstract Euler-Lagrange equations will be formulated in terms of infinitesimal bendings. 
We give various formulations of these Euler-Lagrange equations, one of which shows very clearly the relation 
to the classical Willmore equation (cf. \cite{Weiner}).
In particular, it explains how Euler-Lagrange ordinary differential equations such as those derived in \cite{H-CPAM} 
(or the particular version stated in \cite{StarostinHeijden}) relate to the Willmore equation.
\\
After introducing the
general framework, we study rotationally symmetric surfaces and show
that rotationally symmetric surfaces with finite energy are always stationary.
The simple proof illustrates the use of the stationarity condition introduced
here. (We refer to \cite{H-AnotherRemark} for the analogous result in the context
of von K\'arm\'an theories.)
\\
Moreover, that result allows us to construct a class of smooth metrics $g$ on the unit disk
for which $\W_g$ admits infinitely many distinct stationary points.

\section*{The setting}

Throughout the paper $S\subset\R^2$ is a smoothly bounded simply connected domain,
and $g\in C^{\infty}(\overline{S}, \RR^{2\times 2})$ is a Riemannian metric on $\o S$.
We write $M$ to denote the Riemannian manifold $(S, g)$.
By $\al{\cdot, \cdot}$ resp. $|\cdot|$ we denote the natural scalar product
on bundles inherited from the metric $g$. We set $|g| = \det g$ and $\sqrt{g} = \sqrt{|g|}$.
We use the common convention regarding the raising and lowering of indices. In particular,
$g^{kl}$ denotes the $(kl)$-entry of the matrix $g^{-1}$.
The Christoffel symbols of (the connection associated with) $g$ are denoted by $\Gamma_{ij}^k$, 
that is,
$
\Gamma_{ij}^k := \frac{1}{2} g^{kl}(\d_j g_{il} + \d_i g_{jl} - \d_l g_{ij}).
$
\\
By $TM$ we denote the tangent bundle along $M$; tensor bundles are denoted in the usual way. 
By $\SM$ we denote
the bundle $T^*M\odot T^*M$ of symmetric tensors in $T^*M \otimes T^*M$.
The metric connection is denoted by $D$, and the same letter denotes the natural connection
on tensor bundles. The class of all $L^2$ sections of a bundle $\Gamma$
are denoted by $L^2(\Gamma)$, and similar notation is used for other function spaces.
For $Y\in C^{\infty}(TM)$ we set define $\Def Y = \LL_Y g$ (where $\LL_Y$ is the Lie-derivative),
which is the section of $\SM$ given in coordinates by
$$
(\Def Y)_{\a\b} = \frac{1}{2}\left( (D_{\a}Y)_{\b} + (D_{\b}Y)_{\a}\right).
$$
The standard connection on $\R^3$ is denoted by $\D$. Let $u : M\to\R^3$ be an isometric immersion.
Its normal is denoted by $n$, so
$$
n = \frac{\d_1 u\times\d_2 u}{|\d_1 u\times\d_2 u|}.
$$
Here and in what follows we denote by $\times$ the cross product in $\R^3$.
By $A$ we denote the second fundamental form of $u$;
following common convention its coordinates are denoted by $h_{ij}$,
so
$$
h_{ij} = -\d_i n \cdot\d_j u = n\cdot\d_i\d_j u.
$$
The section $B$ of $\End(TM)$ associated with
a section $b$ of $\SM$ is defined via $\al{BX, Y} = b(X, Y)$, and viceversa.
By $S$ we denote the Weingarten map. It is the negation of the section of $\End(TM)$
associated with $A$. 
\\
By $J$ we denote the natural almost complex structure on $M$, i.e., the section
of $\End(TM)$ determined by the condition that
$$
\D_{JX}u = n\times \D_X u \mbox{ for all }X\in C^{\infty}(TM).
$$
In coordinates we have
$$
J_{\a\b} =
\begin{cases}
-\sqrt{g} &\mbox{ if }(i, j) = (1, 2)
\\
\sqrt{g} &\mbox{ if }(i, j) = (2, 1)
\\
0 &\mbox{ otherwise.}
\end{cases}
$$
We define an associated section $J$ of $\End(\SM)$, also denoted $J$, by setting
$$
(Jq)(X,Y) := q(JX, JY)
$$
for every section $q$ of $\SM$.
Here and in what follows, unless specified otherwise, 
$X$ and $Y$ always denote smooth tangent vector fields. 
We observe that, in coordinates,
$$
(Jq)^{\a\b} = \frac{1}{|g|}\left( \cof q \right)_{\a\b}.
$$
As usual, for sections $b$ of $\SM$ we define the $1$-form $\div b$ by
$$
(\div b)(X) = (D_{\a} b)(\d^{\a}, X).
$$
Here and elsewhere we use the summation convention.
\\
A Codazzi tensor is a section $b$ of $\SM$ satisfying the Codazzi-Mainardi equations, i.e.,
$$
\div (Jb) = 0.
$$
The formal adjoint of $\Def$ is $-\div$. So $b\in L^1_{loc}(\SM)$ is (weakly) Codazzi
if
$$
\int_M \al{Jb, \Def Y} = 0 \mbox{ for all }Y\in C^{\infty}_0(TM).
$$
For a given immersion $u : M\to\R^3$ and for displacements
$\tau$, $\rho : M\to\R^3$ we define the section
$d\tau\cdot d\rho$ of $T^*M\otimes T^* M$ by
$$
(d\tau\cdot d\rho)(X, Y) = \D_X\tau\cdot\D_Y\rho.
$$
On the right-hand side a simple dot denotes the standard scalar product in the ambient
space $\R^3$, so $\D_X\tau\cdot\D_Y\rho = \sum_{i = 1}^3 (\D_X\tau_i)(\D_X\rho_i)$.
\\
The expression $d\tau\times d\rho$ is defined similarly.
\\
By $d\tau\odot d\rho = d\rho\odot d\tau$ we denote the symmetrisation of $d\tau\cdot d\rho$, 
i.e., the section of $\SM$ given by
$$
(d\rho\odot d\tau) = \frac{1}{2}\left( d\tau\cdot d\rho + d\rho\cdot d\tau \right).
$$
A displacement $\tau$ is called an infinitesimal bending of an immersion $u : M\to\R^3$ 
provided that $du\odot d\tau = 0$.
\\
The Hessian of a scalar function $f$ on $M$ is denoted by $\Hess f$. By definition, 
it is the section of $\End(TM)$ (resp. of $\SM$) given by 
$$
\Hess f = Ddf.
$$
Here and elsewhere we
abuse notation by writing $df$ to denote both the gradient vector field of $f$
and its associated one-form. We write $u_* V = \D_V u$ to denote the usual push-forward.
\\
If $F$ and $G$ are matrices of the same order, we will write
$$
F : G = \Tr (F^T G).
$$

\section{Minimisers of $\W_g$}\label{Existence}

As compared to the un-constrained
Willmore functional, the isometry constraint simplifies matters drastically 
when it comes to proving
existence of minimizers. This is mainly because it breaks the invariance under
diffeomorphisms. Therefore, proving existence of minimizers of $\widetilde{\W}_g$ 
will turn out to be straightforward. 
In fact, we will see that the stationary points of $\widetilde{\W}_g$ agree with those of the functional
$\W_g$ defined by
\begin{equation}
\label{D-2}
{\W}_g(u) =
\begin{cases}
\frac{1}{2}\int_M |A|^2 &\mbox{ if }u\in W_g^{2,2}( S, \RR^3)
\\
+ \infty &\mbox{ otherwise.}
\end{cases}
\end{equation}
For all $u\in W^{2,2}_g(S)$ we clearly have
\begin{equation}
\label{hess-2}
|A|^2 = 4H^2 - 2K,
\end{equation}
where $K$ is the Gauss curvature of the metric $g$.
As $K$ is an intrinsic quantity
it is the same for all $u\in W^{2,2}_g(S)$. This is also true for
the boundary integral $\int_{\d S} \kappa_g$.
Hence the stationary points of $\widetilde{\W}_g$ agree with those of the functional
$\W_g$ defined by \eqref{D-2}.
\\
Writing $A^{0} = A - Hg$ for the trace-free part of $A$ and using \eqref{hess-2}, we see
that, up to a constant prefactor and an additive constant, $\W_g$ agrees with the functional
\begin{equation}
\label{wg0} 
{\W^{(0)}}_g(u) =
\begin{cases}
\frac{1}{2}\int_M |A^{0}|^2 &\mbox{ if }u\in W_g^{2,2}( S, \RR^3)
\\
+ \infty &\mbox{ otherwise.}
\end{cases}
\end{equation}
\\

{\bf Remark.} The functional $\W_g$ is coercive with respect to the $W^{2,2}$ seminorm.
\\

This follows from the equality
$
|\D_X\D_Y u|^2 = |A(X, Y)|^2 + |D_X Y|^2
$
and the fact the the second term is intrinsic and therefore the same for all
$u\in W^{2,2}_g(S)$.

In view of this remark, the existence of minimizers under typical boundary conditions
follows at once from the direct method of the calculus of variations.
\begin{proposition}
\label{existence}
The restriction of the functional $\W_g$ to
the space
$$
\AA_0 = \{u\in W_g^{2,2}(S) : \int_M u = 0\}
$$
attains a global minimum on this space.
\\
If $\Lambda\subset\d M$ has positive length and $u_0\in W_g^{2,2}(S)$
is given then similar statements are true with $\AA_0$ replaced by
$$
\AA_{\Lambda, u_0}  = \{u\in W^{2,2}_g(S) : (u, du) = (u_0, du_0)
\mbox{ on }\Lambda\}
$$
or by
$$
\tilde{\AA}_{\Lambda, u_0}  = \{u\in W^{2,2}_g(S) : u = u_0
\mbox{ on }\Lambda\}.
$$
\end{proposition}
\begin{proof}
Sequences in $\AA_0$ with bounded $\W_g$-energy subconverge weakly
in $W^{2,2}$: The second derivatives are controlled because $\W_g$ is coercive with respect to the $W^{2,2}$ seminorm,
the first derivatives are uniformly bounded by the isometry constraint, and by
the normalization (resp. boundary conditions)
together with Poincar\'e's
inequality, $u$ itself is controlled as well.
\\
Moreover, $\W_g$ is obviously lower semicontinuous with
respect to weak $W^{2,2}$-convergence. Finally, notice that the isometry constraint
and the subsidiary conditions are continuous with respect to this convergence.
\end{proof}

\section{Infinitesimal bendings}\label{IBs}

We begin with a definition of infinitesimal bendings which makes sense under minimal regularity assumptions:

\begin{definition}
\label{def-IB} 
A vector field $\tau\in L^2(S, \RR^3)$ is called an {\em infinitesimal bending} of
an immersion $u\in W^{2,2}(S, \RR^3)$ if it satisfies
\begin{equation}
\label{ib-00}
\d_i (\tau\cdot\d_j u) + \d_j (\tau\cdot\d_i u) = 2\tau\cdot\d_i\d_j u \mbox{ in $\DD'(S)$ for }i, j = 1, 2.
\end{equation}
An infinitesimal bending is said to be trivial if it is the velocity field of a rigid motion.
\end{definition}

Observe that if $\tau\in W^{1,2}(S, \RR^3)$ then by the Leibniz rule \eqref{ib-00} is equivalent to
$du\odot d\tau = 0$. In fact, we will mainly encounter infinitesimal bendings of $u$ which belong to the space $W^{2,2}$, so
we will mostly use this condition as a definition of infinitesimal bending.
\\
For a given infinitesimal bending $\tau\in W^{1,2}(M, \RR^3)$ of $u$,
as in \cite{H-Velcic} we define its linearised second fundamental form $b$ by
$b = n\cdot\Hess \tau$.
In coordinates, this reads
$$
b_{ij} = n\cdot (\d_i\d_j\tau - \Gamma^k_{ij}\d_k\tau).
$$
Observe that $b$ as defined here is well-defined as a distribution, because $n\in W^{1,2}$
and $\Hess\tau\in H^{-1}$. The infinitesimal bendings $\tau$ considered in this paper
will mostly belong at least to $W^{2,1}$, 
in which case the definition of $b$ makes sense pointwise almost everywhere 
and $b\in L^1(\SM)$.
\\
There are various ways of representing infinitesimal bendings. Here we will
use three of them. Among geometers, these are well-known in the 
smooth setting, cf. e.g. \cite{Vekua}.
Our approach here differs from the usual one, as we need to display the 
links between these representations. For later use, we formulate these facts
in Sobolev spaces and derive formulae which allow to switch representations
with minimal loss of integrability.

\subsection{Representation via the displacement field}\label{Displacement}

Decomposing the vector field $\tau : S\to\RR^3$ into its tangential part and its
normal part,
\begin{equation}
\label{vf}
\tau = \D_{V}u + \Phi n,
\end{equation}
we have
$
du\odot d\tau = \Def V - \Phi A.
$
Hence $\tau$ of the form \eqref{vf} is an infinitesimal bending of $u$
precisely if 
\begin{equation}
\label{ib}
\Def V = \Phi A.
\end{equation}

As an immediate consequence of these observations and of Definition \ref{def-IB}, 
we note the following result:
\begin{lemma}
\label{vple2} 
Let $\tau\in L^2(M, \RR^3)$, set $\Phi = \tau\cdot n$ and define $V$ by 
$\D_V u = \tau - \Phi n$. Then 
$\tau$ is an infinitesimal bending of $u$ if and only if
$V$ and $\Phi$ satisfy $\Def V = \Phi A$ in the distributional sense.
\end{lemma}

Whenever $f\in W^{1,2}$ and $\mu\in H^{-1}$
then we interpret $f\mu$ as the distribution acting on test functions $\p$
via $(f\mu)(\p) = \mu(\p f)$.
The proof of the following two lemmas is straightforward.

\begin{lemma}\label{lele-1} 
If $\rho$, $f\in W^{1,2}(M)$ then, as distributions,
$$
\Hess (f\rho) = f\Hess\rho + \rho\Hess f + 2d\rho\odot df.
$$
\end{lemma}

\begin{lemma}\label{lele-2} 
If $u\in W^{2,2}(M, \R^3)$ then, as distributions,
$$
(\Hess n)(X, Y) = u_*\left( (D_X S)(Y) \right) - \al{SX, SY}n.
$$
\end{lemma}

For a vector field $V$ we regard $DV$ as a section of $\End(TM)$.
For two sections $\Psi$ and $\Xi$ of $\End(TM)$ we define a section $\Psi\cdot\Xi$
of $T^*M\otimes T^*M$ by
$$
(\Psi\cdot\Xi)(X, Y) = \al{\Psi(X), \Xi(Y)}.
$$
By $\Psi\odot\Xi$ we denote the symmetrisation of this.

\begin{lemma}\label{lele3}
Let $u\in W^{2,2}(M, \R^3)$ and $\Phi\in W^{1,2}(M)$ and $V\in W^{1,2}(TM)$. Then,
as distributions,
\begin{align}
\label{lele3-1} 
n\cdot\Hess \D_V u &= D_V A - 2 S\odot DV
\\
\label{lele3-2} 
n\cdot \Hess (\Phi n) &= \Hess\Phi - \Phi S\odot S.
\end{align}
In particular, the map $\tau = u_*V + \Phi n$ satisfies
\begin{align}
\label{dtauVP}
d\tau &= u_*(DV + \Phi S) + (d\Phi - SV)\otimes n
\\
\label{hesstauVP} 
n\cdot\Hess\tau &= 
D_V A - 2 S\odot DV + \Hess\Phi - \Phi S\odot S.
\end{align}
\end{lemma}
\begin{proof}
Applying Lemma \ref{lele-1} to $0 = \Hess(n\cdot\D_V u)$, we have
\begin{align*}
0 &= \D_V u\cdot \Hess n + n\cdot \Hess\D_V u + 2 dn\odot d( \D_V u ).
\end{align*}
Since by Lemma \ref{lele-2}
\begin{align*}
(\D_V u\cdot\Hess n)(X, Y) &= \al{V,  (D_XS)(Y)}
\\
&= X\left( \al{V, SY} \right) - \al{D_X V, SY} + A(D_X Y, V)
\\
&= -(D_X A)(V, Y) = -(D_V A)(X, Y).
\end{align*}
In the last step we used that $A$ is Codazzi. Since clearly 
$dn\odot d(\D_V u) = S\odot DV$,
equation \eqref{lele3-1} follows.
\\
To prove \eqref{lele3-2} apply Lemma \ref{lele-1} to find
$$
n\cdot\Hess(\Phi n) = \Hess\Phi + \Phi n\cdot\Hess n,
$$
because $n\cdot (d\Phi\odot dn) = 0$. Hence \eqref{lele3-2} follows from Lemma \ref{lele-2}.
\end{proof}

\subsection{Representation via the bending field}\label{Rotation}

The field $\Omega$
in the following lemma is called the bending field of $\tau$.

\begin{lemma}
\label{claim1}
Let $\tau\in W^{1,2}(M, \RR^3)$ be an infinitesimal bending of $u$. Then there exists 
a unique $\Omega\in L^2(M, \RR^3)$ such that
\begin{align}
\label{com-0}
d\tau &= \Omega\times d u \mbox{ almost everywhere.}
\end{align}
It is given by
$\Omega = \D_{\omega} u + \psi n,$
where $\omega\in L^2(TM)$, $\psi\in L^2(M)$ are defined by
\begin{align}
\label{le5-1}
J\omega &= n\cdot d\tau
\\
\label{le5-2}
\psi\ \al{X, Y} &= \D_X \tau\cdot\D_{JY} u.
\end{align}
Moreover, $\Omega$ satisfies $d(\Omega\times du) = 0$
in the distributional sense.
\\
Conversely, if $\Omega\in L^2(M, \RR^3)$ satisfies $d(\Omega\times du) = 0$,
then \eqref{com-0} admits a solution $\tau\in W^{1,2}(M, \RR^3)$,
and $\tau$ is an infinitesimal bending of $u$.
\end{lemma}
\begin{proof}
The last statement is immediate, because $M$ is simply connected, and because
for $\tau$ satisfying \eqref{com-0} clearly $du \odot d\tau = 0$.
\\
In order to prove the first part of the lemma, observe that uniqueness follows at once 
from \eqref{com-0} because $u$ is an immersion.
In order to construct $\Omega$, define $\psi$ and $\omega$
by \eqref{le5-2}, \eqref{le5-1}, and set $\Omega = u_*\omega + \psi n$.
Then clearly $\Omega\in L^2$ since so are $\psi$ and $\omega$, and 
\eqref{com-0} is easily seen to be satisfied.
\\
Finally note that indeed there exists a (clearly unique) solution $\psi$
to \eqref{le5-2}. In fact, we can define $\psi$ by
$$
d\tau\wedge du = 2\sqrt{g}\psi dx^1\wedge dx^2;
$$
on the left we contract in $\R^3$.
Now \eqref{le5-2} follows from
$
d\tau\cdot du = \frac{1}{2}\ d\tau\wedge du,
$
which is true because $\tau$ is an infinitesimal bending.
\end{proof}

\begin{lemma}
\label{claim2}
Let $\psi\in W^{1,1}(M)$, $\omega\in W^{1,1}(TM)$ and set $\Omega = u_*\omega + \psi n$.
Then $\Omega\in W^{1,1}(M, \RR^3)$, and the following assertions are equivalent:
\begin{enumerate}[(i)]
\item \label{curlind} We have $d(\Omega\times du) = 0$ in distributions.
\item \label{com-1} The section $d\Omega\times du$ of $T^* M\otimes T^* M\otimes\R^3$ 
is symmetric.
\item \label{claim2-iii} We have almost everywhere $n\cdot d\Omega = 0$ and $\Tr d\Omega = 0$,
where we view $d\Omega$ as a section of $\End(TM)$, which we may do by the first equality.
\item We have, pointwise almost everywhere,
\begin{align}
\label{com-3}
d\psi &= S\omega
\\
\label{com-5}
\div\omega &= 2H\psi.
\end{align}
\item There exists an infinitesimal bending $\tau\in W^{2,1}(S, \RR^3)$ of $u$ satisfying 
$d\tau = \Omega\times du$.
\end{enumerate}
If any of the above assertions is satisfied, then 
\begin{equation}
\label{dOmega}          
d\Omega = u_*(D\omega + \psi S).
\end{equation} 
\end{lemma}
\begin{proof}
Observe that by the Leibniz rule and since $W^{1,1}$ embeds into $L^2$, 
we have $\Omega\in W^{1,1}$. Clearly
\begin{equation}
\label{dOmega-1}
d\Omega = d(u_*\omega + \psi n) = u_*(D\omega + \psi S) + (d\psi - S\omega)n
\mbox{ a.e. on }M.
\end{equation} 
The Leibniz rule shows that \eqref{curlind} is equivalent to \eqref{com-1}.
The tangential part of \eqref{com-1} is equivalent to $n\cdot d\Omega = 0$.
By \eqref{dOmega-1}, this is just $d\psi = S\omega$, 
and \eqref{dOmega} follows as well.
Next we multiply \eqref{com-1} by $n$ to find $\Tr d\Omega = 0$,
which is equivalent to \eqref{com-5} in view of \eqref{dOmega}.
The existence of an infinitesimal bending $\tau\in W^{1,2}$ solving
$d\tau = \Omega\times du$ is ensured by Lemma \ref{claim1}, and it follows from the Leibniz rule 
that $\Omega\in W^{1,1}$ implies $\tau\in W^{2,1}$.
\end{proof}

{\bf Remark.} If the Gauss curvature $K$ differs from zero on $\o M$,
then the Weingarten map $S$ is invertible. Denoting by $S^{-1}\in L^2(\End(TM))$ 
its fibrewise inverse, we see that the system \eqref{com-3}, \eqref{com-5} is equivalent 
to the conjunction of 
\begin{equation}
\label{wein}
\div (S^{-1}d\psi) - 2H\psi = 0
\end{equation}
with the algebraic equation
\begin{equation}
\label{com-4}
\omega = S^{-1}(d\psi).
\end{equation}

\begin{lemma}
\label{omegab}
Let $\tau\in W^{1,2}(M, \RR^3)$ be an infinitesimal bending of $u\in W^{2,2}_g(S)$,
denote by $\Omega\in L^2(M, \RR^3)$ its
bending field, by $b = n\cdot\Hess\tau$ its linearised second fundamental 
form and by $B$ the section of $\End(TM)$ associated with $b$. Then we have
\begin{equation}
\label{omegab-w} 
\tau\in W^{2,1}(M, \RR^3)\ \iff\ \Omega\in W^{1,1}(M, \RR^3).
\end{equation} 
If these are satisfied, then
\begin{align}
\label{omegab-1} 
(\Hess\tau)(X, Y) &= \D_X\Omega\times \D_Y u + A(X, Y)\ (\Omega\times n)
\end{align}
almost everywhere.
In particular (since $\D_X\Omega$ is tangential),
\begin{align}
\label{dOmb}
\D_X\Omega &= -u_*(JBX)
\\
\label{omegab-2} 
b(X, Y)n &= \D_X\Omega\times \D_Y u\mbox{, i.e., }
b(X, JY) = \D_X\Omega\cdot\D_Y u.
\end{align}
Moreover, writing $\Omega = u_*\omega + \psi n$, we have $\omega$, $\psi\in W^{1,1}$ and
\begin{equation}
\label{claim2-b} 
B = D(J\omega) + \psi J\circ S.
\end{equation}
\end{lemma}
\begin{proof}
By Lemma \ref{claim2} we know that $d\Omega$ is tangential.
Formula \eqref{omegab-w} follows from $d\tau = \Omega\times du$.
From this we also deduce \eqref{omegab-1}. By the definition of $b$,
this implies \eqref{omegab-2}, which in turn is just \eqref{dOmb}.
\\
If $\Omega\in W^{1,1}$, then the Leibniz rule shows that $\psi = \Omega\cdot n$
is $W^{1,1}$ and that $\omega\in W^{1,1}$. Formula \eqref{claim2-b} follows from
\eqref{omegab-2}.
\end{proof}

\begin{corollary}
\label{lemma4}
If $\tau\in W^{2,1}(S, \RR^3)$ is an infinitesimal bending of $u$ then
we have $n\cdot\Hess\tau = 0$ almost everywhere if and only if there exist
$c_0, c_1\in\RR^3$ such that
$\tau = c_0 + c_1\times u.$
\end{corollary}
\begin{proof}
If $\tau = c_0 + c_1\times u$, then clearly
$
n\cdot\Hess\tau = 0.
$
Conversely, if $b = n\cdot\Hess\tau = 0$, then \eqref{omegab-2}
implies that the tangential component of $d\Omega$ is zero. But $d\Omega$ is tangential,
so $\Omega$ is constant.
\end{proof}

\subsection{The linearised Gauss-Codazzi-Mainardi system}\label{LGCM} 

The infinitesimal change $b$ of the second 
fundamental form of an immersion $u$ under a bending clearly satisfies
the linearisation of the Gauss-Codazzi-Mainardi system:
The Codazzi-Mainardi equations are linear, so $b$ is Codazzi.
The linearisation of the
Gauss equation under bendings is $\al{JA, b} = 0$. In coordinates, the linearised
Gauss-Codazzi-Mainardi system is this:
\begin{align}
\label{lg}
b_{11} h_{22} + h_{11} b_{22} - 2 h_{12}b_{12} &= 0
\\
\label{lpc-1}
\d_2 b_{11} - \d_1 b_{12} - \Gamma_{12}^k b_{k1}
+ \Gamma_{11}^k b_{k2}  &= 0
\\
\label{lpc-2}
\d_2 b_{12}  - \d_1 b_{22} - \Gamma^k_{22} b_{k1} + \Gamma^k_{12} b_{k2} &= 0.
\end{align}

\begin{proposition}\label{proneu} 
Let $u\in W^{2,2}_g(S)$. Then the following are true:
\begin{enumerate}[(i)]
\item \label{proneu-1} If $\tau\in W^{2,1}(M, \R^3)$ is an infinitesimal bending of $u$ and 
$b = n\cdot\Hess\tau \in L^2(\SM)$ then $b$ is Codazzi and satisfies $\al{JA, b} = 0$.
\item \label{proneu-2} If $b\in L^2(\SM)$ is Codazzi and satisfies $\al{JA, b} = 0$, 
then there exists a solution 
$\tau\in W^{2,1}(M, \R^3)$ of $n\cdot\Hess\tau = b$. The map $\tau$ is an infinitesimal 
bending of $u$,
and it is unique up to trivial infinitesimal bendings.
\end{enumerate}
\end{proposition}

The proof uses the following key lemma. In its statement,
$d\rho_B$ is the exterior derivative of $\rho_B$ and
$d^D$ denotes the covariant-exterior derivative acting on $\End(TM)$.

\begin{lemma}\label{leproneu} 
Let $u\in W^{2,2}_g(S)$, let $B\in L^2(\End(TM))$ be symmetric and define the $\R^3$-valued
$1$-form $\rho_B$ by setting
$
\rho_B(X) = u_*(JBX).
$
Then
$$
(d\rho_B)(X, Y) = u_*\left(J(d^D B)(X, Y)\right) + \left( A(X, JBY) - A(Y, JBX) \right)n
$$
as distributions. In particular, $\rho_B$ is closed if and only if
the quadratic form $b\in L^2(\SM)$ associated with $B$ is (weakly) Codazzi 
and satisfies $\al{JA, b} = 0$ almost everywhere.
\end{lemma}
\begin{proof}
We compute
\begin{align*}
d\rho_B(X, Y) 
&=
\D_X \D_{JBY} u - \D_Y \D_{JBX} u + u_*\left( JBD_Y X - JBD_XY \right)
\\
&= \left( A(X, JBY) - A(Y, JBX) \right)n 
\\
&+ u_*\left( D_XJBY - D_YJBX - JB[X, Y] \right),
\end{align*}
which is the claim because $D_XJ = JD_X$.
\end{proof}

\begin{proof}[Proof of Proposition \ref{proneu}]
To prove \eqref{proneu-1} note that
Lemma \ref{omegab} implies that $\Omega\in W^{1,1}$ and $\D_X\Omega = - u_*(JBX)$.
Hence, with $\rho_B$ as in Lemma \ref{leproneu} we have $\rho_B = -d\Omega$.
Hence $\rho_B$ is closed, so Lemma \ref{leproneu} implies that $b$
is Codazzi and satisfies $\al{JA, b} = 0$.
\\
To prove \eqref{proneu-2} let $B\in L^2(\End(TM))$ be induced by $b$,
and let $\rho_B$ as in Lemma \ref{leproneu}. Then the lemma implies that $\rho_B$
is closed, so since $M$ is simply connected and since $\rho_B\in L^2$
(because $du\in L^{\infty}$ and $B\in L^2$),
there exists $\Omega\in W^{1,2}(M, \R^3)$
such that $d\Omega = -\rho_B$, i.e., $\D_X\Omega = - u_*(JBX)$.
Clearly $\Omega$ is unique up to a constant vector, and
it satisfies Lemma \ref{claim2} \eqref{claim2-iii}.
Hence that lemma shows that there exists $\tau\in W^{2,1}$ such that 
$d\tau = \Omega\times du$, and $n\cdot\Hess\tau = b$ e.g. by Lemma \ref{omegab}.
\\
Clearly, for given $\Omega$, the map $\tau$ is unique up to a constants;
since the same is true for $\Omega$ itself, this implies that $\tau$ is unique
up to trivial infinitesimal bendings.
\end{proof}

{\bf Remark.} For completeness, we note that
the infinitesimal change of the normal vector $n$
under an infinitesimal bending $\tau$ of $u$ with bending field $\Omega$ 
is given by
$
\mu = \Omega\times n.
$
In fact, denoting by a dot the infinitesimal change of a quantity under
the displacement $\tau$, we have
$
0 = (n\cdot du)\dot{ } = \dot n\cdot du + n\cdot d\tau.
$
This conditions and $n\cdot\dot n = 0$ determine $\dot n$ uniquely.
It is easy to check that $\mu$ satisfies these two conditions.
\\
By linearisation of the Gauss and Weingarten equations for $u$ it is easy to see that
the velocity field $\tau$ of a bending of $u$ satisfies
\begin{align}
\label{lGW-1}
\Hess\tau &= A\otimes\mu + b\otimes n
\\
\label{lGW-2}
\D_X\mu &= \D_{SX}\tau - \D_{BX} u,
\end{align}
This remains true if $\tau$ is 
an arbitrary infinitesimal bending; for \eqref{lGW-1} cf. Lemma \ref{omegab},
and for \eqref{lGW-2} refer, e.g., to \cite{H-preprint}.

\section{Stationary points of $\W_g$}

The following definitions are central, as they
provide a natural notion of (possibly non-minimising) stationary points of $\W_g$.

\begin{definition}
For a given immersion $u\in W^{2,2}_g(S)$ we make the following definitions:
\begin{itemize}
\item A bending of $u$ is a strongly continuous 
one-parameter family $\{u_t\}_{t\in (-1,1)}\subset W^{2,2}_g(S)$ with $u_0 = u$,
and which is such that the weak $L^2$-limit
\begin{equation}
\label{der-general}
b = \lim_{t\to 0}\frac{1}{t} (A_t - A)
\end{equation}
exists. Here $A$ denotes the second fundamenal form of $u$ and 
$A_t$ that of $u_t$. 
\item The section $b\in L^2(\SM)$ is called the linearised second fundamental form induced by 
the bending $\{u_t\}_{t\in (-1,1)}$.
\item Any $b\in L^2(\SM)$ induced as in \eqref{der-general} by 
some strongly continuous family $\{u_t\}_{t\in (-1,1)}\subset W^{2,2}_g(S)$ is called a {\em continuable} linearised
second fundamental form for $u$.
\end{itemize}
\end{definition}

Another natural but (without further regularity hypotheses on $u$)
slightly more restrictive notion would be to regard
$\{u_t\}$ as a bending of $u$ if the weak $W^{2,2}$ limit
\begin{equation}
\label{der} 
\tau = \lim_{t\to 0} t^{-1}(u_t - u_0)
\end{equation} 
exists. In this case, $\tau$ is called the infinitesimal bending 
induced by the bending $\{u_t\}_{t\in (-1,1)}$. Any vector field $\tau$ induced in this manner
by a $W^{2,2}$-bending of $u$ is called a continuable infinitesimal $W^{2,2}$-bending of $u$.
\\
Of course the regularity hypotheses chosen in these definitions are somewhat
arbitrary; one may well wish to consider more regular maps.
One may also impose boundary conditions on $u$ and on the admissible bendings 
(and hence on their induced infinitesimal bendings).
\\
Observe that in the case without boundary conditions considered here, a
trivial class of $W^{2,2}$ bendings is given by the rigid motions. The corresponding
continuable infinitesimal bending fields are precisely those whose gradient is of the form 
$\Omega\times du$ for some constant $\Omega\in\RR^3$. Immersions $u$ which only permit such trivial
bendings are called rigid. They are clearly stationary.
\\

{\bf Remarks.} 

\begin{enumerate}[(i)]
\item We adopt the term `continuable' from the review paper
\cite{IvanovaSabitov-1} and other papers on this subject 
(cf. e.g. \cite{KolegaevaFomenko, Isanov, Klimentov}),
in order to highlight the connection to this large body of literature.
\item Clearly, the limit $\tau\in W^{2,2}(S, \RR^3)$ in
\eqref{der} necessarily is an infinitesimal bending of $u$.
Similarly, every continuable linearised second fundamental
form $b\in L^2(\SM)$ is Codazzi and satisfies $\al{JA, b} = 0$ almost everywhere on $M$
(cf. Section \ref{IBs} for details).
\\
However, it is well-known that the converse is false in general, i.e.,
the class of infinitesimal bendings can be strictly larger than the class of
continuable infinitesimal bendings. Nevertheless, when $K > 0$ then
the two classes are known to agree (in the presence of enough regularity), 
cf. \cite{Klimentov, IvanovaSabitov-2, LeMoPa}. A similar result has recently
been obtained in \cite{H-cpde} for nondegenerate intrinsically flat (i.e. $K = 0$) 
surfaces $u$; it is false in the degenerate case when $u$ contains a planar region.
\end{enumerate}

\begin{definition}
\label{def-2}
For a given immersion $u\in W^{2,2}_g(S)$ we say that
\begin{itemize}
\item $u$ is stationary for $\W_g$ provided that 
$$
\frac{d}{dt}\Big|_{t = 0} \WW_g(u_t) = 0\mbox{ for all bendings $\{u_t\}_{t\in (-1,1)}$ of $u$.}
$$
\item $u$ is formally stationary for $\W_g$ provided that 
$\int_M \al{A, b} = 0$
for all Codazzi tensors $b\in L^2(\SM)$ satisfying $\al{JA, b} = 0$ almost everywhere on $M$.
\end{itemize}
\end{definition}

The notion of formal stationarity is motivated by the following remark:

\begin{proposition}
Let $u\in W^{2,2}_g(S)$ and let $\{u_t\}_{t\in (-1,1)}$ be a bending of $u$
inducing the linearised second fundamental form $b$. Then
\begin{equation}
\label{gene-1-general}
\frac{d}{dt}\Big|_{t = 0} \WW_g(u_t) =
\int_M \al{A, b}.
\end{equation}
In particular, $u$ is stationary for $\W_g$ if and only if $\int_M \al{A, b} = 0$
for all continuable linearised second fundamental forms $b$.
So every formally stationary immersion is stationary.
\end{proposition}
\begin{proof}
By the weak $L^2$-convergence \eqref{der-general} we have 
$
\|A_t - A\|_{L^2}\leq Ct,
$
so $A_t\to A$ strongly in $L^2$, hence
$A + A_t\to 2A$ strongly in $L^2$. Hence using \eqref{der-general}, we conclude that, as $t\to 0$,
$$
\frac{1}{t}\int_M |A_t|^2 - |A|^2
= \frac{1}{t}\int_M \al{A_t + A, A_t - A}
\to 2\int_M \al{A, b}.
$$
\end{proof}

The next lemma follows from a simple computation.

\begin{lemma}\label{le:01} 
If $q$, $b\in L^2(\SM)$ and $\al{Jq, b} = 0$ then $\al{q, b} = (\Tr q)(\Tr b)$.
\end{lemma}

\begin{theorem}\label{pr:01} 
If $u\in W^{2,2}_g(S)$ then the following are equivalent:
\begin{enumerate}[(i)]
\item The immersion $u$ is formally stationary.
\item We have $\int_M H \Tr b = 0$
for all Codazzi tensors $b\in L^2(\SM)$ satisfying $\al{JA, b} = 0$.
\item \label{pr:01-3} We have $\int_M H \Tr b = 0$ 
for all $b\in L^2(\SM)$ with $\div b = 0$ and $\al{A, b} = 0$.
\item \label{pr:01-4} There exist sequences of Lagrange multipliers $\la^{(n)}\in C_0^{\infty}(M)$ and $Y^{(n)}\in C^{\infty}_0(TM)$
such that
\begin{equation}
\label{larub}
\la^{(n)}A + \Def Y^{(n)}\weak Hg
\end{equation}
weakly in $L^2(\SM)$.
\end{enumerate}
\end{theorem}
\begin{proof}
The equivalence of the first three items follows from Lemma \ref{le:01}
and the fact that $\Tr b = \Tr (Jb)$.
The equivalence of \eqref{pr:01-3} and \eqref{pr:01-4} follows by standard functional analysis from the fact that 
$-\Def$ is the formal adjoint of the divergence operator on sections of $\SM$,
cf. e.g. \cite{H-PRSE, H-preprint} for details.
\end{proof}

\begin{corollary}\label{cor01} 
Let $u\in W^{3,1}_g(S)$ be formally stationary.
Then the following are true:
\begin{enumerate}[(i)]
\item \label{cor01-1} We have 
$
\int_M \D_V (H^2) + \Phi H(4H^2 - 2K) + H\Delta\Phi = 0
$
for all $V\in W^{2,2}(TM)$, $\Phi\in W^{2,2}(M)$ satisfying $\Def V = \Phi A$,
\item \label{cor01-2} We have
$
\int_M H\div(J\omega) = 0
$
for all $\omega\in W^{1,2}(TM)$ such that there exists $\psi\in W^{1,2}(M)$ with
\begin{align*}
d\psi &= S\omega
\\
\div \omega &= 2H\psi.
\end{align*}
\item \label{cor01-3} If $K\neq 0$ then
$
\int H\div(JS^{-1}d\psi) = 0
$
whenever $\psi\in W^{1,2}(M)$ satisfies
$
\div (S^{-1}d\psi) = 2H\psi.
$
Here $S^{-1}$ is the section of $\End(TM)$ obtained by inverting $S$ fibrewise.
\end{enumerate}
\end{corollary}
\begin{proof}
If $u$, $V$ and $\Phi$ are as in the hypotheses then the map
$
\tau = u_*V + \Phi n
$
is an infinitesimal bending of
$u$, cf. Lemma \ref{vple2}. Moreover, $b = n\cdot\Hess\tau\in L^2$, according to Lemma \ref{lele3}. 
Proposition \ref{proneu} \eqref{proneu-1} shows that
$b = n\cdot\Hess\tau$ is Codazzi with $\al{JA, b} = 0$. Hence statement \eqref{cor01-1}
follows from Theorem \ref{pr:01} and because 
$$
\Tr b = 2 \D_V H + \Phi |A|^2 + \Delta\Phi,
$$
due to \eqref{hesstauVP}.
\\
If $u$, $\omega$ and $\psi$ are as in the hypotheses then the map
$
\Omega = u_*\omega + \psi n
$
belongs to $W^{1,1}$. 
Moreover, $\Omega$ is the bending field of some infinitesimal bending $\tau$ of
$u$, cf. Lemma \ref{claim2}. Proposition \ref{proneu} \eqref{proneu-1} shows that
$b = n\cdot\Hess\tau$ is Codazzi with $\al{JA, b} = 0$. 
Note that \eqref{claim2-b} implies that $b\in L^2$.
Hence statement \eqref{cor01-2}
follows from Theorem \ref{pr:01} and because 
$
\Tr b = \div\left( J\omega \right),
$
due to \eqref{claim2-b}.
\\
Statement \eqref{cor01-3} follows from \eqref{cor01-2}.
\end{proof}

Observe that the conclusions of Corollary \ref{cor01} are not merely
consequences of formal stationarity
of $u$, but in fact they are roughly equivalent to it. Moreover,
Lagrange multiplier rules as in \eqref{larub} can be easily derived
from Corollary \eqref{cor01}
\\

Clearly, the absolute minimisers
of $\W_g$ when $K = 0$ are the affine maps. This remains true for formally stationary points:

\begin{remark}
Assume that $g$ has constant Gauss curvature $K = 0$. 
Then $u\in W^{2,2}_g(S)$ is formally stationary if and only if $u$ is affine.
\end{remark}
\begin{proof}
If $K = 0$ then $\al{JA, A} = 0$. Since, moreover, $A\in L^2$ is Codazzi,
it is an admissible test tensor in the definition of formal stationarity,
which therefore implies that $\int_M |A|^2 = 0$.
\end{proof}

\begin{remark}
Assume that $g$ has constant Gauss curvature $K = K_0 > 0$. 
Then the absolute minimiser of $\W_g$ is the standard immersion of $(S, g)$ as a subset of the
sphere of radius $K_0^{-1}$.
\end{remark}
\begin{proof}
This follows from the fact that the absolute minimiser of $\W_g$ agrees with that of
the functional \eqref{wg0}, and that surfaces consisting of umbilical points are (subsets of) round spheres.
\end{proof}

\subsection*{Remark about related variational problems}

As seen earlier, up to addition of a term depending only on the metric $g$, the functional $\W_g$ agrees
with the Willmore functional $\int H^2$, restricted to $W^{2,2}_g$. 
Hence, as observed e.g. in \cite{EfratiSharonKupferman-PRE},
minimal immersions are absolute minimisers. So if the metric $g$ is induced by some minimal surface,
then this minimal surface is an absolute minimiser of $\W_g$. 
More generally, if $u\in W^{2,2}_g(S)$ is a stationary point of the classical
Willmore functional (i.e., without the isometry constraint), then it clearly is also 
stationary for $\W_g$.
In this sense, Willmore surfaces are critical points of $\W_g$. Note, however, that a 
`Willmore surface' in this context has a boundary and must satisfy certain boundary conditions.

\section{Symmetric immersions}

The notion of symmetry used in this chapter is common in the context of
nonlinear wave maps, cf. \cite{Struwe-CPAM}. We only use it as a convenient
way to express rotational symmetry.

\subsection{Principle of symmetric stationarity}

Let $G$ be a Lie group acting transitively on $M$ and isometrically on $\R^3$. 
For each $\rho\in G$ define $\la_{\rho} : M\to M$ by $\la_{\rho}(x) = \rho x$.
The action of $\rho$ in $\R^3$ is denoted by $L_{\rho}\in SO(3)$.
An immersion $u : M\to\R^3$ is said to be symmetric if
$$
u\circ\la_{\rho} = L_{\rho}u\mbox{ for all }\rho\in G.
$$
We are only interested in the case when $S = B_1$ is the unit ball and $\D\la_{\rho}\in SO(2)$.
\\
Denote by $\mu_G$ the Haar measure on $G$ normalised such that $\mu_G(G) = 1$. For 
an immersion $u: M\to\R^3$
define $(u)_G : M\to\R$ by
$$
(u)_G(x) = \int_G L_{\rho}^{-1} u(\la_{\rho}(x))\ d\mu_G(\rho).
$$
Clearly, $u$ is symmetric precisely if $u = (u)_G$. 
\\
A section $q$ of $\SM$ is said to be invariant if $\la^*_\rho q = q$ for all $\rho\in G$;
here $\la_{\rho}^*q$ denotes the pullback under the diffeomorphism $\la_{\rho}^*$.
We define
$$
(q)_G = \int_G \la_{\rho}^*q\ d\mu_G.
$$
The following lemmas are readily verified:
\begin{lemma}
If $u$ is symmetric then $A$ and $g$ are invariant in the sense that $\la_{\rho}^* A = A$ 
and $\la_{\rho}^*g = g$ (i.e., $\la_{\rho}$ is an isometry) for all $\rho\in G$. 
In particular, if $b\in L^2(\SM)$ is Codazzi then so is $\la_{\rho}^* b$.
\end{lemma}

\begin{lemma}\label{quinv2} 
Let $q\in L^2(T^*M\otimes T^*M)$ and let $\la : M \to M$ be an isometry.
Then we have
$
J(\la^*q) = \la^*(Jq).
$
In particular, if $q\in L^2(\SM)$ is invariant then so is $Jq$.
\end{lemma}

\begin{lemma}\label{quinv1} 
Let $q$, $b\in L^2(\SM)$ and suppose that $q$ is invariant. Then 
$$
\al{q, (b)_G} = 
\int_G \al{q, b}\circ\la_{\rho}\ d\mu_G(\rho).
$$
\end{lemma}
\begin{proof}
Since $\la_{\rho}$ is an isometry, we have
$
\al{q, b}\circ\la_{\rho} = \al{\la_{\rho}^* q, \la_{\rho}^* b} = \al{q, \la_{\rho}^* b},
$
because $\la_{\rho}^*q = q$. Integration over $\rho$ yields the claim.
\end{proof}

\begin{proposition}\label{quinv} 
Assume that $u\in W^{2,2}_g(S)$ is symmetric and satisfies
$\int H \Tr b = 0$ for all {\em invariant} Codazzi tensors $b\in L^2(\SM)$ with $\al{JA, b} = 0$.
Then $\int H \Tr b = 0$ for all Codazzi tensors $b\in L^2(\SM)$ with $\al{JA, b} = 0$.
\end{proposition}
\begin{proof}
Let $b\in L^2(\SM)$ be a (possibly non-invariant) Codazzi tensor satisfying $\al{JA, b} = 0$.
Since $A$ is invariant, Lemma \ref{quinv2} implies that so is $JA$. 
Hence Lemma \ref{quinv1} implies that $\al{JA, (b)_G} = 0$. 
And Lemma \ref{quinv1} ensures that $(b)_G$ is still Codazzi. Hence the hypotheses imply
that
$
\int_M H\Tr (b)_G = 0.
$
On the other hand, since each $\la_{\rho}$ is an isometry,
\begin{align*}
\Tr(b)_G &= \int_G \Tr (\la_{\rho}^*b)\ d\mu_G(\rho) 
= \int_G (\Tr b)\circ\la_{\rho}\ d\mu_G(\rho).
\end{align*}
Hence by Fubini and since (by invariance of $A$) we have $H = H\circ\la_{\rho}$,
\begin{align*}
0 &= \int_M H\Tr (b)_G = 
\int_M H \left( \int_G (\Tr b)\circ\la_{\rho}\ d\mu_G(\rho) \right)
\\
&= \int_G \left( \int_M (H \Tr b)\circ\la_{\rho} \right)\ d\mu_G(\rho).
\end{align*}
Since $\la_{\rho}$ is an isometry,
the inner integral equals $\int_M H\Tr b$ for all $\rho\in G$. So 
indeed $\int_M H\Tr b = 0$.
\end{proof}

\subsection{Radially symmetric surfaces}

\begin{lemma}\label{symle2} 
Let $T > 0$ and let $R\in W^{1,\infty}(0, T)$ be positive on $[0, T]$,
and let
$$
g(t, \p) = (dt)^2 + R^2(t) (d\p)^2
$$
be a Riemannian metric on $U = [0, T]\times [0, 2\pi]$;
set $M = (U, g)$. Let $b_{ij}\in L^2(0, T)$ and let
$b\in L^2(\SM)$ be given by
$$
b(t, \p) = b_{11}(t) (dt)^2 + b_{22}(t) (d\p)^2 + 2b_{12}(t) dt\odot d\p.
$$
Then $b$ is Codazzi if and only if there exists a constant $C\in\R$ such that
\begin{align}
\label{symle2-1} 
\left( \frac{b_{22}}{R} \right)' &= R' b_{11}
\\
\label{symle2-2} 
b_{12} &= \frac{C}{R}.
\end{align}
\end{lemma}
\begin{proof}
All Christoffel symbols of $g$ are zero, except
$$
\Gamma_{22}^1(t) = - R(t)R'(t)\mbox{ and }\Gamma_{12}^2(t) = (\log R)'(t);
$$
here we set $x_1 = t$ and $x_2 = \p$. Since, moreover, $b_{ij}$ are independent
of $\p$, the Codazzi equations read
\begin{align*}
b_{12}' + \left( \log R \right)' b_{12} &= 0
\\
b_{22}' - RR' b_{11} - \left( \log R \right)' b_{22} &= 0.
\end{align*}
The first of these equations is clearly equivalent to \eqref{symle2-2}.
Dividing the second equation by $R$, we see that it is equivalent to \eqref{symle2-1}.
\end{proof}

\begin{proposition}
Let $T > 0$ and let $R\in C^0([0, T])$ be positive and let $L\in C^0([0, T])$ be such that
\begin{equation}
\label{urad} 
u(x) = R(|x|) \frac{x}{|x|} + L(|x|)e_3
\end{equation} 
defines a map $u\in W^{2,2}(B_T)$. Then $u$ is formally stationary for $\W_g$,
where $g = u^* g_{\R^3}$.
\end{proposition}
\begin{proof}
After possibly reparametrising the curve $t\mapsto (R(t), L(t))$, we may assume that
$(R')^2 + (L')^2 = 1$. Denote by $M$ the Riemannian manifold $(B_T, g)$. Clearly $u$ satisfies
$u(Qx) = Qu(x)$ for all $Q\in SO(2)$ and all $x\in B_T$; on the right-hand side
we regard $SO(2)$ as a subset of $SO(3)$ in the obvious way.
\\
Hence by Proposition \ref{quinv} we must prove that $\int_M \al{A, b} = 0$
for all $SO(2)$-invariant Codazzi tensors $b\in L^2(\SM)$ with $\al{JA, b} = 0$.
Let $b$ be such a tensor.
\\
We introduce radial coordinates $(t, \p)$ via $x = te^{i\p}$; we identify $\mathbb{C}$ with
$\R^2\times\{0\}$. In these coordinates we have (with the usual abuse of notation)
$$
u(t, \p) = R(t)e^{i\p} + L(t) e_3.
$$
Since $(R')^2 + (L')^2 = 1$, we see that
$$
g(t, \p) = (dt)^2 + R^2(t) (d\p)^2.
$$
Invariance implies that there
exist functions $b_{ij}\in L^2(0, T)$ such that
$$
b(t, \p) = b_{11}(t) (dt)^2 + b_{22}(t) (d\p)^2 + b_{12}(t) (dt\odot d\p).
$$
Since $b$ is Codazzi, Lemma \ref{symle2} shows that $b_{22}\in W^{1,2}$ and
\begin{align}\label{sympro-1} 
\left( \frac{b_{22}}{R} \right)' = R' b_{11}.
\end{align}
On the other hand, 
\begin{align*}
h_{11} &= R' L'' - L' R''
\\
h_{12} &= 0
\\
h_{22} &= R L'.
\end{align*}
Hence $\al{JA, b}$ means that
\begin{equation}\label{sympro-3} 
b_{11} R L' = (L'R'' - R' L'') b_{22}.
\end{equation} 
We therefore 
deduce from \eqref{sympro-1} that there exists a constant $C_1\in\R$ such that
\begin{equation}
\label{sympro-2}
b_{22}L' = C_1 R.
\end{equation}
However, since $R'\in L^{\infty}$ and $b_{11}\in L^2$, equation \eqref{sympro-1}
shows, in particular, that $R^{-1}b_{22}$ is bounded. On the other hand, $u\in W^{2,2}$
implies that $L'(0) = 0$. Hence \eqref{sympro-2} implies that $b_{22} = 0$.
But then \eqref{sympro-1} shows that $b_{11}R' = 0$, and \eqref{sympro-3} shows
that $b_{11}L' = 0$. Hence $b_{11} = 0$. Since both $g$ and $h$
are diagonal, this shows that $\al{A, b} = 0$.
\end{proof}

\begin{lemma}
\label{symle3}
Let $\t u$ and $u$ be of the form \eqref{urad}.
Then $\t u$ and $u$ are isometric if and only if (with the obvious
notation) $\t R = R$ and $|\t L'| = |L'|$ almost everywhere on $(0, T)$.
\end{lemma}
\begin{proof}
The metric induced by $u$ is
$$
\left( (R')^2 + (L')^2 \right) (dt)^2 + R^2 (d\p)^2,
$$
and similarly for $\t u$. So keeping in mind that $R$, $\t R$ are nonnegative,
we see that $u$ and $\t u$ are isometric if and only if
$$
\t R = R \mbox{ and }|\t L'| = |L'|\mbox{ almost everywhere,}
$$
because the former implies that $\t R' = R'$ almost everywhere.
\end{proof}

\begin{corollary}\label{symcor1}
Let $\t u$ and $u$ be of the form \eqref{urad}, and assume that they are isometric. Then 
the modulus of their mean curvatures agrees almost everywhere. In particular,
$\W_g(\t u) = \W_g(u)$.
\end{corollary}
\begin{proof}
As before, we assume without loss of generality that $(L')^2 + (R')^2 = 1$.
Set $\kappa = R'L'' - L'R''$ and denote the corresponding quantities for $\t u$
by a tilde.
We compute that $2H = \kappa + L'/R$. Hence
\begin{equation}
\label{symcor1-4} 
4H^2 = \kappa^2 + \frac{(L')^2}{R^2} + \frac{2\kappa L'}{R}.
\end{equation}
By Lemma \ref{symle3}, the second term is clearly the same for $\t u$. 
Regarding the last term, we compute
\begin{equation}
\label{symcor1-5} 
\kappa L' = \frac{1}{2} R' \left( (L')^2\right)' - (L')^2 R''.
\end{equation} 
Hence Lemma \ref{symle3} shows that the last term in \eqref{symcor1-4} also is the same for $\t u$.
By \eqref{symcor1-5} so is $|\kappa| |L'|$. Hence $|\kappa| = |\t\kappa|$ almost everywhere
on $\{L'\neq 0\} = \{\t L' \neq 0\}$. But on $\{L' = 0\} = \{\t L' = 0\}$
we have $L'' = \t L'' = 0$ almost everywhere, hence $\t\kappa = \kappa = 0$.
We conclude that $|\t\kappa| = |\kappa|$ almost everywhere. This shows that
$|H| = |\t H|$ almost everywhere.
\end{proof}

\subsection*{A pathological example}

In this example we construct a class of Riemannian metrics $g\in C^{\infty}(\o B_1)$
such that the functional $\W_g$, with $g = u^* g_{\R^3}$,
admits infinitely many stationary points. This is analogous to the examples in \cite{H-AnotherRemark}.
\\
Let $\eta\in C^{\infty}(\R)$ be nonnegative and supported in 
$(-\frac{1}{2}, \frac{1}{2})$ (but not identically zero).
Let $R\in (0, 1]$ and let
$(t_n)_{n = 1}^{\infty}\subset (0, R)$ be a strictly increasing sequence with 
$t_n\uparrow R$. Set $t_0 = 0$ and define $L:[0,1) \to \R$ by setting
$$
L(t) = \sum_{n = 0}^{\infty} (t_{n+1} - t_n)^n\ \eta\left( \frac{2t - t_n - t_{n+1}}{2(t_{n+1} - t_n)}\right).
$$
Since $L\in W^{2,\infty}(0, 1)$ vanishes near zero, we see that 
$u : B_1\to\R^3$ given by $u(x) = x + L(|x|)e_3$ belongs to $W^{2,\infty}$; 
in fact $u\in C^{\infty}(\o B_1)$,
because $L\in C^{\infty}([0, 1])$.
For each $n = 1, 2, 3, ...$ define $u_n : [0, 1)\to\R$ by
$$
L_n(t) = 
\begin{cases}
-L(t) &\mbox{ if }t\in (t_n, t_{n+1})
\\
L(t) &\mbox{ otherwise,}
\end{cases}
$$
and define $u_n\in W^{2,\infty}(B_1)$ by setting $u_n(x) = x + L_n(|x|)e_3$.
Since $|L_n'| = |L'|$ everywhere, Lemma \ref{symle3} shows that 
all $u_n$ are isometric to $u$. 
Clearly, the $u_n$ are pairwise distinct and of the form \eqref{urad}.
Proposition \ref{quinv} shows that each $u_n$ is (even formally) stationary for
$\W_g$, with $g = u^* g_{\R^3}$.

\vspace{1cm}

\def\cprime{$'$}


\begin{thebibliography}{10}

\bibitem{BauerKuwert}
M.~Bauer and E.~Kuwert.
\newblock Existence of minimizing {W}illmore surfaces of prescribed genus.
\newblock {\em Int. Math. Res. Not.}, (10):553--576, 2003.

\bibitem{bohle}
C.~Bohle, G.~P. Peters, and U.~Pinkall.
\newblock Constrained {W}illmore surfaces.
\newblock {\em Calc. Var. Partial Differential Equations}, 32(2):263--277,
  2008.

\bibitem{EfratiSharonKupferman}
E.~Efrati, E.~Sharon, and R.~Kupferman.
\newblock Elastic theory of unconstrained non-euclidean plates.
\newblock {\em J. Mech. Phys. Solids}, 57:762--775, 2009.

\bibitem{EfratiSharonKupferman-PRE}
E.~Efrati, E.~Sharon, and R.~Kupferman.
\newblock Hyperbolic non-euclidean elastic strips and almost minimal surfaces.
\newblock {\em Phys. Rev. E}, 83, 2013.

\bibitem{fjm1}
G.~Friesecke, R.~D. James, and S.~M{\"u}ller.
\newblock A theorem on geometric rigidity and the derivation of nonlinear plate
  theory from three-dimensional elasticity.
\newblock {\em Comm. Pure Appl. Math.}, 55(11):1461--1506, 2002.

\bibitem{FJMM-cras}
Gero Friesecke, Richard~D. James, Maria~Giovanna Mora, and Stefan M{\"u}ller.
\newblock Derivation of nonlinear bending theory for shells from
  three-dimensional nonlinear elasticity by {G}amma-convergence.
\newblock {\em C. R. Math. Acad. Sci. Paris}, 336(8):697--702, 2003.

\bibitem{GemmerVenka-PhysicaD}
J.~Gemmer and S.~Venkataramani.
\newblock Shape selection in non-euclidean plates.
\newblock {\em Physica D: Nonlinear Phenomena}, 240(19):1536--1552, 2011.

\bibitem{GemmerVenka}
J.~Gemmer and S.~Venkataramani.
\newblock Shape transitions in hyperbolic non-euclidean plates.
\newblock {\em Soft Matter}, 9(34):8151--8161, 2013.

\bibitem{Helfrich}
W.~Helfrich.
\newblock Elastic properties of lipid bilayers: theory and possible
  experiments.
\newblock {\em Z. Naturforsch. A}, C28:636--703, 1973.

\bibitem{H-CPAM}
P.~Hornung.
\newblock Euler-{L}agrange equation and regularity for flat minimizers of the
  {W}illmore functional.
\newblock {\em Comm. Pure Appl. Math.}, 64(3):367--441, 2011.

\bibitem{H-preprint}
P.~Hornung.
\newblock The {W}illmore functional on isometric immersions.
\newblock {\em Preprint}, 2012.

\bibitem{H-cpde}
P.~Hornung.
\newblock Continuation of infinitesimal bendings on developable surfaces and
  equilibrium equations for nonlinear bending theory of plates.
\newblock {\em Comm. PDE}, 38:1368 -- 1408, 2013.

\bibitem{H-AnotherRemark}
P.~Hornung.
\newblock Another remark on constrained von {K}\'arm\'an theories.
\newblock {\em arXiv:1410.3806v1}, 2014.

\bibitem{H-PRSE}
P.~Hornung.
\newblock A remark on constrained von {K}\'arm\'an theories.
\newblock {\em Proc. R. Soc. A}, 470, 2014.

\bibitem{H-Velcic}
P.~Hornung and I.~Vel{\v c}i{\'c}.
\newblock Derivation of a homogenized von-{K}{\'a}rm{\'a}n shell theory from 3d
  elasticity.
\newblock {\em Ann. Inst. Henri Poincare (C)}, 2014.

\bibitem{Isanov}
T.~G. Isanov.
\newblock The continuation of infinitesimal bendings.
\newblock {\em Dokl. Akad. Nauk SSSR}, 234(6):1257--1260, 1977.

\bibitem{IvanovaSabitov-1}
I.~Ivanova-Karatopraklieva and I.~Kh. Sabitov.
\newblock Deformation of surfaces. {I}.
\newblock In {\em Problems in geometry, {V}ol.\ 23 ({R}ussian)}, Itogi Nauki i
  Tekhniki, pages 131--184, 187. Akad. Nauk SSSR Vsesoyuz. Inst. Nauchn. i
  Tekhn. Inform., Moscow, 1991.
\newblock Translated in J. Math. Sci. {{\bf{7}}0} (1994), no. 2, 1685--1716.

\bibitem{IvanovaSabitov-2}
I.~Ivanova-Karatopraklieva and I.~Kh. Sabitov.
\newblock Bending of surfaces. {II}.
\newblock {\em J. Math. Sci.}, 74(3):997--1043, 1995.
\newblock Geometry, 1.

\bibitem{Klimentov}
S.~B. Klimentov.
\newblock Extension of higher-order infinitesimal bendings of a simply
  connected surface of positive curvature.
\newblock {\em Mat. Zametki}, 36(3):393--403, 1984.

\bibitem{KolegaevaFomenko}
E.~M. Kolegaeva and V.~T. Fomenko.
\newblock Continuation of infinitesimal bendings of surfaces to analytic
  bendings under external constraints.
\newblock {\em Mat. Zametki}, 45(2):30--39, 141, 1989.

\bibitem{KupfermanSolomon}
R.~Kupferman and J.~P. Solomon.
\newblock A {R}iemannian approach to reduced plate, shell, and rod theories.
\newblock {\em J. Funct. Anal.}, 266(5):2989--3039, 2014.

\bibitem{KS-flow}
E.~Kuwert and R.~Sch{\"a}tzle.
\newblock Gradient flow for the {W}illmore functional.
\newblock {\em Comm. Anal. Geom.}, 10(2):307--339, 2002.

\bibitem{KS-annals}
E.~Kuwert and R.~Sch{\"a}tzle.
\newblock Removability of point singularities of {W}illmore surfaces.
\newblock {\em Ann. of Math. (2)}, 160(1):315--357, 2004.

\bibitem{KS-conformal}
E.~Kuwert and R.~Sch{\"a}tzle.
\newblock Minimizers of the {W}illmore functional under fixed conformal class.
\newblock {\em arXiv:1009.6168v1}, 2008.

\bibitem{LeMoPa}
M.~Lewicka, M.~G. Mora, and M.~R. Pakzad.
\newblock The matching property of infinitesimal isometries on elliptic
  surfaces and elasticity of thin shells.
\newblock {\em Arch. Ration. Mech. Anal.}, 200(3):1023--1050, 2011.

\bibitem{PinkallSterling}
U.~Pinkall and I.~Sterling.
\newblock Willmore surfaces.
\newblock {\em Math. Intelligencer}, 9(2):38--43, 1987.

\bibitem{Riviere-willmore}
T.~Rivi{\`e}re.
\newblock Analysis aspects of {W}illmore surfaces.
\newblock {\em Invent. Math.}, 174(1):1--45, 2008.

\bibitem{Riviere-L2}
T.~Rivi{\`e}re.
\newblock Variational principles for immersed surfaces with {$L^2$}-bounded
  second fundamental form.
\newblock {\em Preprint}, 2010.

\bibitem{Schatzle-CV}
R.~Sch{\"a}tzle.
\newblock The {W}illmore boundary problem.
\newblock {\em Calc. Var. Partial Differential Equations}, 37(3-4):275--302,
  2010.

\bibitem{Schygulla}
J.~Schygulla.
\newblock Willmore minimizers with prescribed isoperimetric ratio.
\newblock {\em Arch. Ration. Mech. Anal.}, 203(3):901--941, 2012.

\bibitem{SharonMarderSwinney}
E.~Sharon, M.~Marder, and H.~Swinney.
\newblock Leaves, flowers and garbage bags: Making waves.
\newblock {\em American Scientist}, 92(3):254, 2004.

\bibitem{Simon-willmore}
L.~Simon.
\newblock Existence of surfaces minimizing the {W}illmore functional.
\newblock {\em Comm. Anal. Geom.}, 1(2):281--326, 1993.

\bibitem{StarostinHeijden}
E.~L. Starostin and G.~H.~M. van~der Heijden.
\newblock The shape of a {M}\"obius strip.
\newblock {\em Nature materials}, 6:563--567, 2007.

\bibitem{Struwe-CPAM}
M.~Struwe.
\newblock Equivariant wave maps in two space dimensions.
\newblock {\em Comm. Pure Appl. Math.}, 56(7):815--823, 2003.

\bibitem{Vekua}
I.~N. Vekua.
\newblock {\em Verallgemeinerte analytische {F}unktionen}.
\newblock Herausgegeben von Wolfgang Schmidt. Akademie-Verlag, Berlin, 1963.

\bibitem{Weiner}
J.~L. Weiner.
\newblock On a problem of {C}hen, {W}illmore, et al.
\newblock {\em Indiana Univ. Math. J.}, 27(1):19--35, 1978.

\bibitem{Willmore-book}
T.~J. Willmore.
\newblock {\em Riemannian geometry}.
\newblock Oxford Science Publications. The Clarendon Press Oxford University
  Press, New York, 1993.

\end{thebibliography}
\end{document}